\documentclass[superscriptaddress,showpacs,10pt]{revtex4}
\usepackage{amssymb}
\usepackage{amsmath}
\usepackage{amsthm}
\usepackage{graphicx}

\begin{document}%

\newcommand{\ket}[1]{|#1\rangle}
\newcommand{\bra}[1]{\langle#1|}
\newcommand{\inner}[2]{\langle#1|#2\rangle}
\newcommand{\lr}\longrightarrow
\newcommand{\ra}\rightarrow
\newcommand{\tr}{{\rm Tr}}
\newcommand{\sgn}{{\rm sgn}}
\newcommand{\fsp}{{\rm span}}
\newcommand{\fsup}{{\rm supp}}
\newcommand{\fdg}{{\rm diag}}
\newcommand{\norm}[1]{\parallel#1\parallel}
\newcommand\e{\varepsilon}

\newtheorem{thm}{Theorem}
\newtheorem{Prop}{Proposition}
\newtheorem{Coro}{Corollary}
\newtheorem{Lemma}{Lemma}
\newtheorem{Def}{Definition}

\title{Randomized Algorithms and Lower Bounds for Quantum Simulation}
\author{Chi Zhang}
\email{cz2165@columbia.edu} \affiliation{Department of Computer
Science, Columbia University, New York, USA, 10027}
\date{\today}

\begin{abstract}
We consider deterministic and {\em randomized} quantum algorithms
simulating $e^{-iHt}$ by a product of unitary operators
$e^{-iA_jt_j}$, $j=1,\dots,N$, where $A_j\in\{H_1,\dots,H_m\}$,
$H=\sum_{i=1}^m H_i$ and $t_j > 0$ for every $j$. Randomized
algorithms are algorithms approximating the final state of the
system by a mixed quantum state. First, we provide a scheme to bound
the trace distance of the final quantum states of randomized
algorithms. Then, we show some randomized algorithms, which have the
same efficiency as certain deterministic algorithms, but are less
complicated than their opponentes. Moreover, we prove that both
deterministic and randomized algorithms simulating $e^{-iHt}$ with
error $\e$ at least have $\Omega(t^{3/2}\e^{-1/2})$ exponentials.
\end{abstract}


\maketitle

\section{Introduction}

While the computational cost of simulating many particle quantum
systems using classical computers grows exponentially with the
number of particles, quantum computers nonetheless have the
potential to carry out the simulation efficiently \cite{Kassal}.
This property, pointed out by Feynman, is one of the founding ideas
of the field of quantum computation. The simulation problem is also
related to quantum walks and adiabatic optimization
\cite{adia_1,adia_2,walk_1,walk_2,walk_3,walk_4}.

A variety of quantum algorithms have been proposed to predict and
simulate the behavior of different physical and chemical systems. Of
particular interest are {\it splitting methods} that simulate the
unitary evolution $e^{-iHt}$, where $H$ is the system Hamiltonian,
by a product of unitary operators of the form $e^{-iA_jt_j}$,
$j=1,\dots,N$, where $A_j\in\{H_1,\dots,H_m\}$, $H=\sum_{i=1}^m H_i$
and assuming the $H_i$ do not commute.

A recent paper \cite{Berry} shows that high order splitting methods
\cite{Suzuki,Suzuki_2} can be used to derive bounds for $N$ that are
asymptotically tight. This work also provides algorithms that
achieve the upper bounds. However, the derived algorithms require
some of the $t_j$ to be negative, which may limit their application.
For instance, the algorithms cannot be used for the simulation of
diffusion operators, because there exists no inverse exponential
diffusion operator, as noted by Suzuki \cite{Suzuki} who proposed
the high order splitting methods. The reason is that for spiting
methods with order of convergence greater than or equal to three,
some of the $\{t_j\}$ must be negative \cite{Suzuki_2}.

In this paper, we consider deterministic and {\em randomized}
quantum algorithms simulating $e^{-iHt}$ using only positive
$\{t_j\}$. By randomized algorithms we mean algorithms approximating
the final state of the system by a mixed quantum state. We show
that:
\begin{enumerate}
\item
The increase in the trace distance of the quantum states in a
randomized algorithm is bounded from above by
$$2\norm{E(U_\omega) - U_0} + E(\norm{U_\omega-U_0}^2),$$
where $U_0$ is the unitary evolution being simulated, $U_\omega$
denotes the randomized algorithm and $E(\cdot)$ is the expectation.
\item
Deterministic and randomized algorithms simulating $e^{-iHt}$ by
approximating it by $\prod_{j=1}^N e^{-iA_jt_j}$ with error $\e$
must satisfy
$$N=\Omega(t^{3/2}\e^{-1/2}).$$
\item
The optimal deterministic algorithm is based on the
Baker-Campbell-Housdorf formula \cite{NC}.
\item
An optimal randomized algorithm is obtained by a direct application
of the Trotter formula \cite{NC}.
\end{enumerate}

\section{Randomized Algorithms for Quantum Simulation}

Let us now state the problem in more details, then discuss the
algorithms and their performance. A quantum system evolves according
to the Schr\"odinger equation
\begin{equation}\label{SchEqu}
i\frac{d}{dt}\ket{\psi(t)} = H\ket{\psi(t)},
\end{equation}
where $H$ is the system Hamiltonian. For a time-independent
$H$, the solution of the Schr\"odinger is
\begin{equation}\label{sol}
\ket{\psi(t)} = e^{-iHt}\ket{\psi_0},
\end{equation}
where $\ket{\psi_0}$ is the initial state at $t=0$. Here we assume
that $H$ is the sum of local Hamiltonians, i.e.,
\begin{equation}\label{sum}
H = \sum_{k=1}^m H_k,
\end{equation}
and all the $H_k$ are such that $e^{-iH_k\tau}$ can be implemented
efficiently for any $\tau>0$. Therefore, we will be using a product
of the unitary operators $U=\prod_{j=1}^N e^{-iA_jt_j}$, where
$A_j\in\{ H_1,\dots,H_m\}$, $t_j>0$, to simulate $U_0=e^{-iHt}$.
However, since the $H_k$ do not commute in general, it introduces an
error in the simulation. We measure this error using the trace
distance, as in \cite{Berry}. Out goal is to obtain tight bounds for
$N$ for algorithms achieving accuracy $\e$ in the simulation and to
show optimal algorithms.

Variations of this problem that do not set restrictions on $t_j$
have been extensively studied in the literature, see, e.g.,
\cite{Berry, Suzuki, NC, Zalka}. As far as we know only
deterministic algorithms have been considered. In this paper, we
propose a randomized model for simulating quantum systems, which
simplifies the design of algorithms without compromising their
efficiency.

In the randomized model, the sequence of unitary operators is
selected randomly according to a certain probability distribution.
The distribution can be realized either by \lq\lq coin-flips\rq\rq\
or by \lq\lq control qubits\rq\rq, which requires some ancillary
qubits. As a result, the algorithm is a product of a random sequence
of unitary operators $U_\omega =\prod_{j=1}^{N_\omega}
e^{-iA_{j,\omega}t_j}$ selected with probability $p_\omega$. Hence,
the final state of the quantum algorithm is the mixed state
\begin{equation}
\rho = \sum_{\omega} p_\omega U_{\omega}\ket{\psi_0}
\bra{\psi_0}U_\omega^\dagger
\end{equation}
For more general cases, where the input state of the simulation is
not exactly $\ket{\psi_0}$, but a different mixed state $\rho_0$,
the final state is $\rho = \sum_{\omega} p_\omega U_{\omega}\rho_0
U_\omega^\dagger$.

In order to analyze the efficiency of randomized algorithms, we show
an upper bound for the trace distance between the desired final
state and the one computed by a randomized algorithm.

\begin{Lemma}
Let $U_0$ be the unitary evolution being simulated by a set of
random unitary evolutions $U_\omega$ as we described above. Then the
trace distance between $\sigma =
U_0\ket{\psi_0}\bra{\psi_0}U_0^\dagger$ and $\rho$ is bounded from
above by
\begin{equation}\label{measure}
\begin{split}
&D(\rho_0,\ket{\psi_0}\bra{\psi_0})+2\norm{\sum_{\omega} p_\omega
U_\omega-U_0}+\sum_{\omega} p_\omega \norm{U_\omega-U_0}^2\\
=&D(\rho_0,\ket{\psi_0}\bra{\psi_0}) + 2\norm{E(U_\omega) - U_0} +
E(\norm{U_\omega-U_0}^2),
\end{split}
\end{equation}
where $D(\cdot)$ denotes the trace distance and $E(\cdot)$ denotes
the expectation.
\end{Lemma}

\begin{proof}
First, we calculate the difference of the output states $\rho_1$ and
$\rho'_1$, which is
\begin{equation}
\begin{split}
&\sum_\omega p_\omega U_\omega\rho_0U_\omega^{\dagger} - U_0\ket{\psi_0}\bra{\psi_0}U_0^{\dagger}\\
= &\sum_\omega(U_0+U_\omega-U_0)\rho_0(U_0+U_\omega-U_0)^{\dagger} - U_0\ket{\psi_0}\bra{\psi_0}U_0^{\dagger}\\
= &\sum_\omega p_\omega(U_\omega-U_0)\rho_0U_0^{\dagger}+\sum_\omega
p_\omega U_0\rho_0(U_\omega-U_0)+\sum_\omega
p_\omega(U_\omega-U_0)\rho_0(U_\omega-U_0)^{\dagger} + \sum_\omega
p_\omega
U_0\rho_0U_0^{\dagger}-U_0\ket{\psi_0}\bra{\psi_0}U_0^{\dagger}\\ =
&(\sum_\omega p_\omega U_\omega-U_0)\rho_0U_0^{\dagger} +
U_0\rho_0(\sum_\omega p_\omega U_\omega-U_0)^{\dagger}+\sum_\omega
p_\omega(U_\omega-U_0)\rho_0(U_\omega-U_0) +
U_0(\rho_0-\ket{\psi_0}\bra{\psi_0})U_0^{\dagger}.
\end{split}
\end{equation}
Hence,
\begin{equation}
\begin{split}
&D(\rho'_1,\rho_1) =\tr|\rho'_1-\rho_1|\\
\leq &\tr|(\sum_\omega p_\omega U_\omega-U_0)\rho_0U_0^{\dagger}|
+\tr|U_0\rho_0(\sum_\omega p_\omega U_\omega-U_0)^{\dagger}|
+\sum_\omega p_\omega\tr|(U_\omega-U_0)\rho_0(U_\omega-U_0)|
+\tr|\rho_0-\ket{\psi_0}\bra{\psi_0}|\\
\leq &2\norm{\sum_\omega p_\omega U_\omega - U_0} + \sum_\omega
p_\omega\norm{U_\omega-U_0}^2 + \tr|\rho_0-\ket{\psi_0}\bra{\psi_0}|\\
=&D(\rho_0,\ket{\psi_0}\bra{\psi_0})+2\norm{E(U_\omega) - U_0} +
E(\norm{U_\omega-U_0}^2).
\end{split}
\end{equation}
\end{proof}
Let $\Delta D = D(\rho_1,\rho'_1) -
D(\rho_0,\ket{\psi_0}\bra{\psi_0})$, which is the augment of trace
distance. From the above lemma, we know in the simulation, $\Delta
D$ is bounded from above by $\norm{E(U_\omega) - U_0}$ and
$E(\norm{U_\omega-U_0}^2)$. Moreover, it is easy to check that
$\Delta D = \Theta(\norm{E(U_\omega) - U_0})$ for certain $\rho_0$,
as well as $\Delta D = \Theta(E(\norm{U_\omega-U_0}^2))$ for some
other $\rho_0$. Therefore, the lower bound is also tight
asymptotically, i.e.,
$$\Delta D = \Theta(2\norm{E(U_\omega) - U_0} +
E(\norm{U_\omega-U_0}^2)).$$

For the convenience of the reader, below we give two examples of
randomized algorithms and we use the lemma above to analyze their
cost. It turns out that the second algorithm is optimal.

\begin{itemize}

\item {\bf Algorithm 1}

Divide the total evolution time $t$ into equal $K$ small segments of
size $\Delta t$.

Let $\rho_0=\ket{\psi_0}\bra{\psi_0}$ be the input to the first
stage of the algorithm.

Consider the $k$-th stage of the algorithm where the input is
$\rho_{k-1}$, $k=1,\dots,mK$. The algorithm chooses uniformly and
independently at random operators from $\{e^{-iH_1\Delta
t},\dots,e^{-iH_m\Delta t}\}$. Hence, the output of stage $k$ is
$$\rho_k=
\sum_{j=1}^{m} \frac{1}{m} e^{-iH_j\Delta t}\rho_{k-1}e^{iH_j\Delta
t}.$$

The final result of the algorithm is $\rho_{mK}$ and is used to
approximate $\sigma$.
\end{itemize}

Due to Lemma 1, the error of this algorithm for simulating
$e^{-iH\Delta t}$ by $m$ consecutive stages (i.e., by the stages
$km+1$ and $(k+1)m$, for any $k=0,\dots,K-1$) is bounded by two
elements, $\norm{E(U_\omega) - U_0}$ and $E(\norm{U_\omega-U_0}^2)$,
where $U_{\omega}$ is the product of the sequence $m$ operators.
Since the selection in each stage is independent and uniform,
\begin{equation}
E(U_\omega) = (\frac{1}{m}\sum_{j=1}^m H^{-iH_m\Delta t})^m = I -
i\sum_{j=1}^m H_j\Delta t + O(\Delta t^2).
\end{equation}
Hence, $\norm{E(U_\omega) - U_0} = O(\Delta t^2)$. Furthermore, for
any $\omega$, $U_\omega = I + O(\Delta t)$, then
$E(\norm{U_\omega-U_0}^2) = O(\Delta t^2)$. Therefore, the error in
each $m$ consecutive stages is $O(\Delta t^2)$. Thus, the total
error is $\e=O(K\Delta t^2)$ and the total number of exponentials
used is $N = mK = O(t^2/\e)$.

We remark that this is equal modulo a constant to the cost of the
deterministic algorithm that is based on a direct application of the
Trotter formula, i.e., the one that uses
$$\prod_{j=1}^m e^{-iH_j\Delta t}$$
to simulate $e^{-iH\Delta t}$. However, {\bf Algorithm 1} has
certain advantages over this deterministic algorithm. In the
deterministic algorithm, in order to simulate $e^{-iH\Delta t}$, we
need to store the current index $j$ of $e^{-iH_j\Delta t}$, for
$j=1,\cdots, m$. However, in {\bf Algorithm 1}, each stage is
independent and the algorithm is ``memoryless''.
\begin{itemize}
\item {\bf Algorithm 2}

Divide the total evolution time $t$ into equal $K$ small segments of
size $\Delta t$.

Let $\rho_0=\ket{\psi_0}\bra{\psi_0}$ be the input to the first
stage of the algorithm.

Consider the $k$-th stage of the algorithm where the input is
$\rho_{k-1}$, $k=1,\dots,K$. The algorithm select an operator
uniformly and independently at random from the set of operators
$$\{\prod_{j=1}^m e^{-iH_\sigma(j)\Delta t}: \sigma \text{ varies over all permutations on } m \text{ symbols }\}.$$
Then, the output of the $k$-th stage is
$$\rho_k=
\sum_{\sigma} \frac{1}{m!} (\prod_{j=1}^m e^{-iH_\sigma(j)\Delta
t})\rho_{k-1}(\prod_{j=1}^m e^{iH_\sigma(j)\Delta t}).$$

The final result of the algorithm is $\rho_{K}$ and is used to
approximate $\sigma$.
\end{itemize}

Let
\begin{equation}
\begin{split}
U_\sigma &= \prod_{j=1}^m e^{-iH_\sigma(j)\Delta t} \\
&= \prod_{j=1}^m (I-iH_{\sigma(j)}\Delta t - \frac{1}{2}
H^2_{\sigma(j)}\Delta t^2 + O(\Delta t^3))\\
& = I - i\sum_{j=1}^m H_{\sigma(j)} \Delta t -
\frac{1}{2}\sum_{j=1}^m H^2_{\sigma(j)}\Delta t^2
-\sum_{j<k}\sum_{j<k}H_{\sigma(j)}H_\sigma(k)\Delta
t^2 + O(\Delta t^3)\\
& = I - i\sum_{j=1}^m H_j\Delta t - \frac{1}{2}\sum_{j=1}^m
H^2_j\Delta t^2 -\frac{1}{2}\sum_{j<k}H_{\sigma(j)}H_\sigma(k)\Delta
t^2 + O(\Delta t^3).
\end{split}
\end{equation}
In each stage, the simulating operator $U_{\omega}$ could be
$U_\sigma$ with probability $1/m!$, for each $\sigma$ in the
permutations on $m$ symbols. Since, for each $\sigma$,
$\norm{U_\sigma - U_0} = O(\Delta t^2)$, so
$E(\norm{U_\omega-U_0}^2) = O(\Delta t^4)$. Moreover, $E(U_\omega) =
I -i\sum_{j=1}^mH_j\Delta t - \frac{1}{2}(\sum_{j=1}^mH_j)^2 \Delta
t^2 + O(\Delta t^3)$, hence $\norm{E(U_\omega)-U_0} = O(\Delta
t^3)$. Due to Lemma 1, the error of this algorithm for simulating
$e^{-iH\Delta t}$ in each stage is $O(\Delta t^3)$. Thus the total
error $\e=O(K\Delta t^3)$. Hence, for a given $\e$ the value of $K$
in {\bf Algorithm 2} is smaller than that in {\bf Algorithm 1}. The
total number of exponentials used is $N = mK = O(t^{3/2}/\e^{1/2})$.

We remark that this is equal modulo a constant to the cost of a
deterministic algorithm solving the problem. The difference is that
the deterministic algorithm is more slightly more complicated than
the one discussed in the previous item. It is based on the
Baker-Campbell-Housdorf formula (Strang splitting) and uses
$$\prod_{j=1}^m e^{-iH_j\Delta t/2}\prod_{j=m}^{1} e^{-iH_j\Delta t/2}$$
to simulate $e^{-iH\Delta t}$.

\section{Lower Bounds for Randomized Algorithms}

In fact, {\bf Algorithm 2} is asymptotically optimal among all
randomized algorithms simulating the evolution of the quantum
system. Before proving the optimality of {\bf Algorithm 2}, we start
with a lemma.

\begin{Lemma}
For any $0\leq x_i\leq 1$, $i=1,\cdots,N$, and $\sum_{i=1}^N x_i =
2$, let $S$ be the sum of all elements in $\{x_ix_jx_k: i<j<k, 2|k -
i, 2\nmid j - i\}$. Then $S < 1/3$.
\end{Lemma}

\begin{proof}
For $N = 3$, it is easy to check $S \leq (2/3)^3 < 1/3$.

Assume that, for $N<M$, the conclusion holds.

For the case $N=M$, the global minimum of $S$ will be achieved at a
local minimum or the border. Moreover, if the global minimum is
obtained at the border, which means some $x_i = 0$ or $1$, it
reduces to the case $N<M$. Then, we only consider its local minimum
and assume $x_i \neq 0$ for each $i$.

Let $f(x) = S - \lambda(\sum_i^N x_i -2)$. In any local minimum,
$\frac{\partial f}{\partial x_i} = 0$, for each $x_i$. Then
$$\frac{\partial f}{\partial x_i} = \sum_{\substack{i<j<k \\ 2\nmid j-i\\ 2|k-i}} x_jx_k+\sum_{\substack{j<i<k\\ 2\nmid j-i\\ 2\nmid k-i}} x_jx_k
+\sum_{\substack{j<k<i\\ 2| j-i\\ 2\nmid k-i}} x_jx_k-\lambda = 0$$,
and
$$\frac{\partial f}{\partial x_{i+2}} = \sum_{\substack{i+2<j<k \\ 2\nmid j-i\\ 2|k-i}} x_jx_k+\sum_{\substack{j<i+2<k\\ 2\nmid j-i\\ 2\nmid k-i}} x_jx_k
+\sum_{\substack{j<k<i+2\\ 2| j-i\\ 2\nmid k-i}} x_jx_k-\lambda =
0.$$
Combined these two equations, we have
$$x_{i+1}(\sum_{\substack{k\geq i+2\\2|k-i}}x_k + \sum_{\substack{k\leq i-1\\2\nmid k-i}}x_k
- \sum_{\substack{k\geq i+3\\2\nmid k-i}}x_k - \sum_{\substack{k\leq
i\\2|k-i}}x_k) = 0.$$ From the assumption, $x_{i+1}\neq 0$, we have
$$\cdots+x_{i-3}+x_{i-1}+x_{i+2}+x_{i+4}+\cdots = \cdots+x_{i-2}+x_i
+ x_{i+3}+x_{i+5}+\cdots.$$ Then, we consider $\frac{\partial
f}{\partial x_{i+1}} = 0$ and $\frac{\partial f}{\partial x_{i+3}} =
0$, which can derive $$\cdots+x_{i-1}+x_{i+1}+x_{i+4}+x_{i+5}+\cdots
= \cdots+x_{i-2}+x_i + x_{i+3}+x_{i+5}+\cdots.$$ Combine them
together, we have $x_{i+1} = x_{i+2}$. Therefore,
$x_2=x_3=\cdots=x_{N-1}$. Then, by considering $\frac{\partial
f}{\partial x_{1}} = \frac{\partial f}{\partial x_{2}}$ and
$\frac{\partial f}{\partial x_{N}} = \frac{\partial f}{\partial
x_{N-1}}$, we have $x_1 = 0$ when $N$ is even; and
$x_1=x_2=\cdots=x_N$ when $N$ is odd. Since the first case is
contradict to our assumption, we need only consider the case $N$ is
odd. Let $N=2K+1$, then each term in $S$ is $(\frac{2}{N})^3$, and
there are $\frac{1}{6}K(K+1)(2K+1)$ terms. Therefore, at the local
minimum
$$S = \frac{1}{6}K(K+1)(2K+1)(\frac{2}{2K+1})^3 = \frac{1}{3}(1-\frac{1}{N^2}) < \frac{1}{3}.$$
Hence, for any $N$ the conclusion holds.
\end{proof}

From the above Lemmas, we have the following two theorems.

\begin{thm}
For both deterministic or randomized algorithms, the error of
simulating $e^{-iH\Delta t}$ is no less than $\Omega(\Delta t^3)$.
\end{thm}

\begin{proof}
Since deterministic algorithms are special cases of randomized
algorithms, it is enough to consider randomized algorithms. Assume
$e^{-iH\Delta t}$ is simulated by $U_{\omega}$ with probability
$p_{\omega}$. Consider the Hamiltonians $H_1$ and $H_2$, in a given
$U_\omega$, let $\alpha_1\Delta t, \alpha_2\Delta t, \cdots,
\alpha_K\Delta t$ be the total evolution time of $H_1$ between two
consecutive evolution of $H_2$, while $\beta_1\Delta t,
\beta_2\Delta t,\cdots, \beta_{K'}\Delta t$ are the total evolution
time of $H_2$ between two consecutive evolution of $H_1$. For
example, if $$U_\omega = e^{-iH_1\lambda_1\Delta
t}e^{-iH_2\lambda_2\Delta t}e^{-iH_3\lambda_3\Delta
t}e^{-iH_2\lambda_4\Delta t}e^{-iH_1\lambda_5\Delta
t}e^{-iH_4\lambda_6\Delta t}e^{-iH_1\lambda_7\Delta t},$$ then
$\alpha_1 = \lambda_1$, $\alpha_2 = \lambda_5+\lambda_7$ and
$\beta_1=\lambda_2+\lambda_4$. So, it is easy to see $|K-K'|\leq 1$.
Due to Lemma 1, the difference of trace distance is decided by
$E(\norm{U_\omega-U_0}^2)$ and $\norm{E(U_\omega) - U_0}$. If
$\sum_{j=1}^K \alpha_j \neq 1$ or $\sum_{j=1}^{K'} \beta_j \neq 1$
for some $\omega$, $\norm{U_\omega-U_0}^2 = \Omega(\Delta t^2)$,
hence $E(\norm{U_\omega-U_0}^2) = \Omega(\Delta t^2)$. Hence, we
only need to consider the situation $\sum_{j=1}^K \alpha_j = 1$ and
$\sum_{j=1}^{K'} \beta_j = 1$. Let us focus on the terms
$iH_1H_2H_1\Delta t^3$ and $iH_2H_1H_2\Delta t^3$. In $e^{-iH\Delta
t}$, each of which has coefficients $1/6$. If the simulation has an
error less than $O(\Delta t^3)$, the coefficients of them in
$E(U_\omega)$ must also be $1/6$, therefore the sum of them should
be $1/3$. However, we will show in every $U_\omega$, the sum of
these two coefficients is less than $1/3$. Without loss of
generality, assume $K\geq K'$, let $x_{2j-1} = \alpha_j$, for
$j=1,\cdots,K$, and $x_{2j}=\beta_j$, for $j=1,\cdots,K'$. Then, the
coefficient of $iH_1H_2H_1\Delta t^3$ is the sum of $x_jx_kx_l$,
where $j<k<l$, $j,l$ are odd, and $k$ is even, while the coefficient
of $iH_1H_2H_1\Delta t^3$ is the sum of $x_jx_kx_l$, where $j<k<l$,
$j,l$ are even, and $k$ is odd. Since $\sum_{j=1}^{K+K'} x_j =
\sum_{j=1}^K \alpha_j + \sum_{j=1}^{K'} \beta_j = 2$, due to Lemma
2, the sum of the coefficients is always less than $1/3$. Since
there always exists $\Theta(\Delta t^3)$ term in any simulation, the
error of any simulation is no less than $\Omega(\Delta t^3)$.
\end{proof}

Therefore, we have the main theorem.

\begin{thm}\label{them1}
Any deterministic or randomized algorithm simulating $e^{-iHt}$ by
approximating it by $\prod_{j=1}^N e^{-iA_jt_j}$, $t_j>0$, with
error $\e$ satisfies
$$N=\Omega(t^{3/2}\e^{-1/2}).$$
\end{thm}

\begin{proof}
Assume the simulation is comprised of $K$ stages, and in the $j$-th
stage, there are constant exponentials used to simulate
$e^{-iHt_j}$, where $\sum_{j=1}^N = t_j$. From the above theorem,
the final error is $\Omega(\sum_{j=1}^K t_j^3)$, the minimum of
which is $\Omega(\frac{t^3}{K^2})$. Hence, to grantee the final
error is bounded by $\e$, $K=\Omega(t^{3/2}\e^{-1/2})$. Therefore,
$N=\Omega(K)=\Omega(t^{3/2}\e^{-1/2})$.
\end{proof}

From Theorem \ref{them1}, it is straightforward to obtain the
following two Corollaries.

\begin{Coro}
The deterministic algorithm based on the Baker-Campbell-Housdorf
formula (Strang splitting) is asymptotically optimal.
\end{Coro}

\begin{Coro}
The randomized algorithm Algorithm $2$ is asymptotically optimal.
\end{Coro}

\section{Summary}

In summary, we provide the randomized model of quantum simulation,
and provide some randomized algorithms which are easier to implement
than certain deterministic algorithms, but have the same efficiency.
Moreover, we provide a lower bound for quantum simulation, therefore
prove the optimality of the deterministic algorithm based on Strang
splitting and one of our randomized algorithms. Note that, the lower
bound and the optimality is under the assumption $t_j$ is positive
in the simulation. Without this restriction, some algorithms have
faster running time than the lower bound \cite{Berry}. Furthermore,
randomized algorithms also bring certain benefits in this
unrestricted situation. For instance, when $m=2$, to simulate
$e^{-iH\Delta t}$, it needs at least $7$ exponentials to obtain an
error bound $\Theta(\Delta t^3)$ \cite{Suzuki}, however a randomized
algorithm can obtain the same error bound with only $4$ exponentials
\cite{future}.

We are grateful to Anargyros Papageorgiou, Joseph F. Traub, Columbia
University for their very helpful discussions and comments.

\end{document}